\documentclass[12pt,leqno,letterpaper]{article}

\usepackage{amsmath,amsthm,enumerate,amssymb}
\usepackage[latin1]{inputenc}
\usepackage[T1]{fontenc}
\usepackage[english]{babel}
\usepackage{graphicx}

\newtheorem{theorem}{Theorem}[section]

\newtheorem{lemma}[theorem]{Lemma}

\newtheorem{example}[theorem]{Example}

\newcommand{\tr}{{\rm Tr\hskip -0.2em}~}

\newcommand{\Var}{\rm Var}

\begin{document}

\title{WYD-like skew information measures}
\author{Frank Hansen}
\date{July 18 2012}

\maketitle

\begin{abstract}

Following recent advances in the theory of operator monotone functions we introduce new classes of WYD-like skew information measures.
\end{abstract}

\section{Introduction}

Let $ h $ be a positive operator monotone function defined in the positive half-axis. We suppose $ h $ is  normalized such that $ h(1)=1, $ and that $ h $ is neither constant nor the identity function.

Kawasaki and Nagisa \cite{kn:nagisa:2012} recently proved that the functions
\[
t\to\frac{(t-a)(t-b)}{(h(t)-h(a))(h^\sharp(t)-h^\sharp(b))}\qquad t>0
\]
are operator monotone for $ a,b>0, $ 
where $ h^\sharp $ is the operator monotone function defined by setting
\[
h^\sharp(t)=\frac{t}{h(t)}\qquad t>0.
\]
This is an interesting observation since by setting $a=b=1, $ choosing $ h(t)=t^p $ for $ 0<p<1, $ and multiplying with a positive constant we obtain the operator monotone functions
\begin{equation}\label{functions generating WYD}
t\to p(1-p)\frac{(t-1)^2}{(t^p-1)(t^{1-p}-1)}\qquad t>0.
\end{equation}
These functions were first studied by Hasegawa and Petz \cite{kn:hasegawa:1996} in the context of monotone metrics on the state space of a quantum mechanical system. Monotone metrics are generated by positive normalized operator monotone functions $ f $ defined in the positive half-axis and  satisfying the functional equation $ f(t)=tf(t^{-1})=f^*(t) $ for $ t>0. $ We refer to \cite{kn:hansen:2011:2} for a modern account of the theory of operator monotone functions.

To each monotone metric there is an associated metric adjusted skew information which is a measure of quantum information under the presence of a conservation law. The metric adjusted skew informations associated with the monotone metrics generated by the functions in (\ref{functions generating WYD}) are the much celebrated
Wigner-Yanase-Dyson skew informations defined by
\[
I_\rho(p, A)=-\frac{1}{2}\tr [\rho^p,A][\rho^{1-p},A]\qquad 0<p<1,
\]
where $ \rho $ is a state and $ A $ is a conserved quantity.

The theorem of Kawasaki and Nagisa thus points to the existence of additional families of WYD-like skew informations that may be more computationally tractable than general metric adjusted skew informations.

The proof presented by Kawasaki and Nagisa is rather long and depends heavily on the theory of analytic functions. Following an idea in \cite{kn:hansen:2010:1} we shall present a short proof that also entails a generalization of the theorem.

\begin{theorem}\label{Main theorem}
Let $ h $ be a normalized positive function defined in the positive half-axis, and suppose $ h $ is neither constant nor the identity function. Consider the mathematical rule
\begin{equation}\label{main function}
f(t)=\frac{(t-1)^2}{(h(t)-1)(h^\sharp(t)-1)}\qquad t>0
\end{equation}
which may not define $ f $ everywhere in the positive half-axis.

\begin{enumerate}[(i)]

\item If $ h $ is operator monotone then the rule defines $ f $ as a positive operator monotone function in the positive half-axis. 

\item If $ h $ and $ h^\# $ are operator convex then the rule defines  $ f $ as an operator monotone decreasing function in the positive half-axis.

\end{enumerate}

\end{theorem}

\begin{proof} We know under both set of assumptions that $ h $ and $ h^\# $ are continuously differentiable functions, and by differentiation of $ f^\# $ we readily obtain $ h'(1)+(h^\#)'(1)=1. $ For $ t\ne 1 $ the denominator in (\ref{main function}) may be written
\[
\begin{array}{rl}
(h(t)-1)(h^\sharp(t)-1)&=t-h(t)-h^\#(t)+1\\[1.5ex]
&=-\bigl(h(t)-1+h^\#(t)-1-(t-1)\bigr),
\end{array}
\]
and therefore
\[
\frac{(h(t)-1)(h^\sharp(t)-1)}{t-1}=-\left(\frac{h(t)-1}{t-1}+\frac{h^\#(t)-1}{t-1}-1\right).
\]
We may thus write $ f $ on the form
\[
f(t)=-\frac{t-1}{g(t)-1}
\]
where the function
\[
g(t)=\left\{\begin{array}{ll}
                 \displaystyle\frac{h(t)-1}{t-1}+\frac{h^\#(t)-1}{t-1}\quad &t>0,\, t\ne 1\\[2ex]
                 1                                                                     &t=1.
                 \end{array}\right.
\]
is continuous also in $ t=1. $ 

If $ h $ is operator monotone then so is $ h^\# $, see for example \cite[Corollary 4.2]{kn:hansen:2011:2}. We then realize that both $ h $ and $ h^\# $ are operator concave in the positive half-axis, see for example \cite[Corollary 2.2]{kn:hansen:2011:2}. By Bendat and Sherman's theorem \cite{kn:bendat:1955} it follows that the positive function $ g $ is operator monotone decreasing thus operator convex. By once more invoking Bendat and Sherman's theorem we obtain that $ -f(t)^{-1} $ and hence $ f(t) $ is operator monotone.

If $ h $ and $ h^\# $ are assumed operator convex then $ g $ is operator monotone in the positive half-axis by Bendat and Sherman's theorem. As above it follows that  $ g $ is operator concave. By once more invoking Bendat and Sherman's theorem we obtain that $ f(t) $ is well-defined and operator monotone decreasing in the positive half-axis.
\end{proof}

If the function $ h $ in the above theorem is operator monotone then both $ 0<h'(1)<1 $ and $ 0<(h^\#)'(1)<1, $ see for example
\cite[Corollary 4.7]{kn:hansen:2011:2}. In addition $ f(0)=\lim_{t\to 0} f(t)>0. $

Symmetry of a monotone metric \cite{kn:petz:1996:2} is expressed through the functional equation $ f(t)=tf(t^{-1})=f^*(t) $ for the generating operator monotone function.  We also say that a function $ f $ defined in the positive half-axis is symmetric if $ f^*=f. $ It is important to know when functions written on the form (\ref{main function})  are symmetric. The following lemma is a small exercise left to the reader.

\begin{lemma}
Functions $ f $ written on the form (\ref{main function}) are symmetric $ (f=f^*) $ if either $ h=h^* $ or $ h=\tilde h, $ where
$ \tilde h(t)=h(t^{-1})^{-1} $ for $ t>0. $
\end{lemma}

The functions $ h(t)=t^p $ appearing in the WYD-metrics satisfy $ h=\tilde h. $

The positive operator monotone functions defined in the positive half-axis satisfying $ h=h^* $ are characterized in 
\cite[Theorem 1]{kn:hansen:2006:2}, cf. also \cite{kn:hansen:2008:1,kn:hansen:2010:1,kn:hansen:2011:1}.

The positive operator monotone functions defined in the positive half-axis satisfying $ h=\tilde h $ are characterized in \cite[Theorem 1.1]{kn:hansen:1981}.

\section{WYD-like skew information measures}

Let $ c $ be the Morozova-Cencov function of a regular monotone metric. It is given on the form
 \[
c(x,y)=\frac{1}{y f(xy^{-1})}\qquad x,y>0,
 \]
where $ f $ is a positive normalized operator monotone function defined in the positive half-axis such that
$ f^*=f $ and $ f(0)>0. $  

The corresponding {\it metric adjusted skew information} is defined by
\begin{equation}\label{metric adjusted skew information}
I^c_\rho(A) = \frac{m(c)}{2} K^c_\rho(i[\rho, A], i[\rho, A]),
\end{equation}
where $ m(c)=f(0) $ is the metric constant, and $ K_\rho^c $ is the monotone metric generated by $ f, $
cf. \cite{kn:petz:1996:2,kn:hasegawa:1996,kn:hansen:2008:1}.
It may be extended \cite[Theorem 3.8]{kn:hansen:2008:1} to positive semi-definite states $ \rho. $ The metric adjusted skew information may be written on the convenient form
\begin{equation}\label{representation of metric adjusted skew information}
I^c_\rho(A)=\frac{m(c)}{2}\displaystyle\tr A\, \hat c(L_\rho,R_\rho)A,
\end{equation}
where
\begin{equation}\label{c hat}
\hat c(x,y)=(x-y)^2 c(x,y)\qquad x,y>0,
\end{equation}
and $ L_\rho $ and $ R_\rho $ denote left and write multiplication with $ \rho. $ 

It is because of this expression that functions written on the form (\ref{main function}) lead to simplified formulas. Indeed, after normalization of $ f, $ the corresponding function in (\ref{c hat}) takes the form
\[
\hat c(x,y)=\frac{1}{h'(1)(1-h'(1))}\bigl(x+y-yh(xy^{-1})-yh^\#(xy^{-1})\bigr)
\]
expressed additively in terms of the perspectives of $ h $ and $ h^\#. $ If we insert these functions in (\ref{representation of metric adjusted skew information}) we obtain:

\begin{theorem}\label{WYD-like skew informations}
Let $ f $ be a symmetric function given on the form
\[
f(t)=h'(1)(1-h'(1))\frac{(t-1)^2}{(h(t)-1)(h^\sharp(t)-1)}\qquad t>0,
\]
where $ h $ is a normalized operator monotone function. The corresponding metric adjusted skew information may be written
\begin{equation}\label{simplified MASI}
I^c_\rho(A) =\frac{1}{2(1-h(0))}\left(2\tr\rho A^2-\tr\bigl(Ah(\Delta_\rho)A\rho + Ah^\#(\Delta_\rho) A\rho\bigr)\right),
\end{equation}
where $ \Delta_\rho $ is the modular operator given by
\[
\Delta_\rho A=\rho A \rho^{-1}.
\]
\end{theorem}

\begin{example}
If we chose $ h(t)=t^p $ for $ 0<p<1 $ then the expression in (\ref{simplified MASI}) reduces to
\[
\begin{array}{rl}
I_\rho^c(A)&=\displaystyle\frac{1}{2}\left(2\tr\rho A^2-\tr\bigl(A(\Delta_\rho)^pA\rho + A(\Delta_\rho)^{-p}\rho A\bigr)\right)\\[2.5ex]
&=\displaystyle\frac{1}{2}\left(2\tr\rho A^2-\tr\bigl(A\rho^pA\rho^{1-p}+A\rho^{1-p}A\rho^p\bigr)\right)\\[2.5ex]
&=\displaystyle-\frac{1}{2}\tr [\rho^p,A]\cdot[\rho^{1-p},A]
\end{array}
\]
and we recover, as expected, the Wigner-Yanase-Dyson skew information measures \cite{kn:wigner:1963,kn:lieb:1973:1}.

\end{example}

\subsection{Examples of WYD-like skew information measures}

The functions,
\[
f_\alpha(t)=t^\alpha\left(\frac{1+t}{2}\right)^{1-2\alpha}\qquad t>0,
\]
studied in \cite{kn:hansen:2006:2} provide, for $ 0\le\alpha\le 1, $ a monotonous bridge between the largest and the smallest monotone metrics. They satisfies $ f_\alpha^*=f_\alpha $ and $ f_\alpha^\#=f_{1-\alpha} $ for $ 0\le\alpha\le 1. $ We therefore obtain

\begin{theorem}
The formula
\[
I_\rho^\alpha(A)=\tr\rho A^2-\frac{1}{2}\tr\Bigl[ A\Bigl(\frac{1+\Delta_\rho}{2}\Bigr)^{1-2\alpha}\rho^\alpha A \rho^{1-\alpha}+ A\Bigl(\frac{1+\Delta_\rho}{2}\Bigr)^{-1+2\alpha}\rho^{1-\alpha} A \rho^\alpha\Bigr]
\]
defines for $ 0\le\alpha\le 1 $ a metric adjusted skew information.

\end{theorem}

The skew information measures defined in the preceding theorem share all of the general properties of the Wigner-Yanase-Dyson skew information measures, cf. \cite{kn:hansen:2008:1} and \cite[Theorem 1.2]{kn:hansen:2010:1}. In particular, they are non-negative quantities bounded by the variance $ \Var_\rho(A), $ and they are convex functions in the state variable $ \rho. $

Notice that $ I_\rho^{1/2}(A) $ is the Wigner-Yanase skew information.

The simplest operator monotone functions $ f $ satisfying $ f=\tilde f $ are given by $ f(t)=t^p $ for $ 0\le p\le 1. $ It is rather hard to write down any other example although they exist in abundance. The function
\[
f(t)=\left(\frac{t}{1+t}\right)^{1+t}(1+t)^{1+\frac{1}{t}}\qquad t>0
\]
is an example of an operator monotone function satisfying the functional equation $ f(t)=\tilde f(t)=f(t^{-1})^{-1}, $ cf. \cite{kn:hansen:1981}. Therefore, also this function generates, through Theorem~\ref{Main theorem} and Theorem~\ref{WYD-like skew informations}, a WYD-like skew information measure.

\subsection{Unbounded measures of skew information}

It is the condition $ f(0)>0 $ for the operator monotone function generating the underlying monotone metric that makes the associated metric adjusted skew information bounded. A monotone metric with this property $ (f(0)>0) $ is called regular. 

The metric adjusted skew information vanishes if the underlying metric is not regular. It was therefore proposed in \cite{kn:hansen:2010:1} to simply remove the factor $ f(0)/2 $ and define the so-called unbounded metric adjusted skew information
\begin{equation}\label{unbounded metric adjusted skew information}
I^c_\rho(A) =  K^c_\rho(i[\rho, A], i[\rho, A]),
\end{equation}
when the metric is not regular. This may happen when the generating function $ h $  in Theorem~\ref{Main theorem} is operator convex. 

An unbounded metric adjusted skew information shares all the general properties of metric adjusted skew informations including convexity in the state variable. The only exception is boundedness. It can therefore not be extended from the state manifold to the state space.

If we choose $ h(t)=t^p $ for $ -1<p<0 $ or $ 1<p<2 $ then both $ h $ and $ h^\# $ are operator convex and $ h=\tilde h. $ The associated function constructed in (\ref{main function}) takes after normalization the form
\[
f(t)=p(1-p)\frac{(t-1)^2}{(t^p-1)(t^{1-p}-1)}\qquad t>0.
\]
Since the normalization changes the sign we obtain that $ f $ is operator monotone, and since $ f(0)=0 $
it generates the unbounded metric adjusted skew information
\begin{equation}\label{extension of WYD}
I_\rho^c(A)=\frac{-1}{p(1-p)}\tr [\rho^p,A]\cdot[\rho^{1-p},A].
\end{equation}

These unbounded extensions of the WYD-skew information measures were already discussed by Hasegawa \cite[Theorem 3]{kn:hasegawa:1993}, but not in the context of monotone metrics. Later they were studied by
Jen\v{c}ov\'{a} and Ruskai \cite{kn:jencova:2010}, cf. also \cite{kn:hansen:2010:1}.



\end{document}